\documentclass[journal,twoside,web]{ieeecolor}
\usepackage{generic}
\usepackage{cite}
\usepackage{amsmath,amssymb,amsfonts}
\usepackage{algorithmic}
\usepackage{graphicx}
\usepackage{algorithm,algorithmic}
\usepackage{hyperref}
\usepackage{textcomp}
\usepackage{float}
\usepackage{makecell}
\newtheorem{cat}{Category}
\usepackage{theoremref}
\usepackage{glossaries}

\newtheorem{remark}{Remark}
\newtheorem{definition}{Definition}
\newtheorem{theorem}{Theorem}
\newtheorem{lemma}{Lemma}
\newtheorem{proposition}{Proposition}
\def\BibTeX{{\rm B\kern-.05em{\sc i\kern-.025em b}\kern-.08em
T\kern-.1667em\lower.7ex\hbox{E}\kern-.125emX}}
\newcommand{\bbT}{\mathbb{T}}
\newcommand{\bbS}{\mathbb{S}}
\newcommand{\bbX}{\mathbb{X}}
\newcommand{\rom}[1]{\uppercase\expandafter{\romannumeral #1\relax}}
\makeglossaries
\newglossaryentry{domain}{
    name={$\mathbb{T}$},
    description={The domain of the function $f$},
    sort={d}
}
\newglossaryentry{groundset}{
    name={$\mathbb{S}$},
    description={The ground set containing all possible components of strings},
    sort={s}
}

    \newglossaryentry{horizon}{
    name={$K$},
    description={The horizon length for our decision problem},
    sort={h}
}

\newglossaryentry{allowedelts}{
    name={$\mathbb{S}(S)$},
    description={The set of elements $s \in \mathbb{S}$ for which $Ss \in \mathbb{T}$},
    sort={d}
}
\newglossaryentry{optimal}{
    name={$O_L$},
    description={The optimal string of actions for our optimization problem},
    sort={o}
}

\newglossaryentry{greedy}{
    name={$G_K$},
    description={The string of elements obtained by using the greedy algorithm},
    sort={g}
}

\newglossaryentry{increment}{
    name={$\Delta(S_i)$},
    description={Incremental difference of the function after adding a component $u$ to a string $S$},
    sort={d}
}

\newglossaryentry{gamma_G}{
    name={$\gamma_G$},
    description={The curvature constant from Conforti and Cornuejols which requires the computation of increments along the greedy sequence},
    sort={g}
}

\newglossaryentry{gamma_G''}{
    name={$\gamma_G''$},
    description={The curvature constant from Conforti and Cornuejols which assumes the function to be defined on the entire ground set},
    sort={g}
}

\newglossaryentry{curlyineq}{
    name={$\preccurlyeq$},
    description={Prefix notation},
    sort={p}
}

\newglossaryentry{deltastar}{
    name={$\Delta^*(s)$},
    description={The minimum increment obtained by adding in $s$ to a string comprised of elements in $|\mathbb{S}\setminus s|$},
    sort={d}
}

\newglossaryentry{sigma}{
    name={$\sigma_{\mathbb{T}}(s, K)$ \hspace{1em}},
    description={ \hspace{1em} Strings $S$ of length $K-1$ such that $Ss$ belongs to $\mathbb{T}$},
    sort={s}
}

\newglossaryentry{gamma_G'}{
    name={$\gamma_G'$},
    description={The curvature constant from Conforti and Cornuejols which involves computations at the last step of the greedy sequence},
    sort={g}
}

\begin{document}
\sloppy
\title{A Performance Bound for the Greedy Algorithm in a Generalized Class of String Optimization Problems}
\author{Brandon~Van~Over, \IEEEmembership{Graduate Student Member, IEEE}, Bowen~Li, \IEEEmembership{Graduate Student Member, IEEE}, Edwin~K.~P.~Chong, \IEEEmembership{Fellow, IEEE} and Ali~Pezeshki, \IEEEmembership{Senior Member, IEEE}
\thanks{Manuscript submitted on November 26th, 2024 for review.} 
\thanks{This work is supported in part by the AFOSR under award FA8750-20-2-0504 and by the NSF under award CNS-2229469. \textit{(Brandon Van Over and Bowen Li contributed equally to this work.)}}% <-this % stops a space
%B.~Van~Over, B.~Li, E.~K.~P.~Chong, A.~Pezeshki
\thanks{The authors are with the Department of Electrical and Computer Engineering, Colorado State University, Fort Collins, CO 80523, USA (email: b.van\_over@colostate.edu; bowen.li@colostate.edu; \qquad \qquad
edwin.chong@colostate.edu; ali.pezeshki@colostate.edu).
}
}

\maketitle

\begin{abstract}
We present a simple performance bound for the greedy scheme in string optimization problems. Our approach generalizes the family of greedy curvature bounds established by Conforti and Cornuejols (1984). Specifically, we examine three bounds they introduced for evaluating the performance of the greedy scheme in maximizing monotone submodular set functions. We first generalize two of these bounds to string optimization problems in a manner that includes maximizing monotone submodular set functions as a special case. Next, we derive a simpler and more computable bound that applies to a broader class of functions with string domains. We then prove that our bound is superior to two of their bounds and provide a counterexample to show that the third bound is incorrect under the assumptions in Conforti and Cornuejols (1984). We demonstrate our results through two applications. First, we apply our bound to sensor coverage problems with both monotone set and string submodular objective functions. The second application is a social welfare maximization problem involving a monotone non-submodular black-box utility function.

\end{abstract}


\section{Introduction}
\label{sec:introduction}
Many sequential decision-making and machine-learning problems involve optimally choosing a \emph{string} (ordered set also referred to as  \textit{sequence} in \cite{golovin2011adaptive}) of actions  over a finite horizon 
%or a set of elements under certain cardinality constraints 
to maximize a given objective function. String optimization problems have an added complexity relative to \emph{set} optimization problems because the value of the objective function \emph{also} depends on the order of the actions. The dependence on the order creates a drastically larger space of possible solutions, and for sufficiently large spaces makes the problem computationally intractable. Such computational infeasibility motivates approximating the optimal solution by various computable schemes. One of the most common approximation schemes is the \emph{greedy scheme}, where at each step we select the action that maximizes the increment in the objective function. An obvious question arises: How well does the greedy scheme perform for our \emph{string optimization} problem relative to the optimal solution? 

In this paper, we derive a \emph{ratio bound} for the performance of greedy solutions relative to the optional solution. Such bounds have the form
\[
\frac{f(G_K)}{f(O_K)}\geq \beta,
\]
where $f(G_K)$ is the objective function value of the greedy solution, $f(O_K)$ is the value of the optimal solution, and $K$ is the horizon length. To create such a ratio bound, we use the following two steps:

\begin{enumerate}
    \item Find an upper bound $B \geq f(O_K)$.
    \item Divide $f(G_K)$ by $B$ to obtain $\beta = f(G_K)/B \leq f(G_K)/f(O_K)$.
\end{enumerate}
\noindent\textbf{Main contributions:} In \cite{conforti1984submodular}, greedy curvature bounds were presented that have found applications in robotics, machine learning, economics, and many other fields. The main contribution of this paper is ultimately showing that all of the greedy curvature bounds in \cite{conforti1984submodular} are now superseded by our superior performance bound. More specifically, the contributions of this paper are as follows: 
\begin{enumerate}
    \item We prove an extension of the $\alpha_G$ and $\alpha_G''$ bounds in \cite{conforti1984submodular} to general string optimization problems with minimal assumptions. The result is shown in \thref{GenConforti}. 
    \item We introduce a performance bound for the greedy scheme in \thref{lemmatopk} which applies to a more general class of functions that includes monotone submodular functions, and whose string domains generalize the set matroids seen in \cite{conforti1984submodular}. We require that our function $f$ satisfies a single inequality that only involves the optimal string of actions, and we impose a condition on the domain of $f$ that guarantees that we can compare the function values associated with components of the optimal and greedy string. 
    \item Our bound is easily computable, does not involve computations past the horizon, and provides an individualized performance bound for each problem. 
    \item We prove that our bound is better than the $\alpha_G$ bound of \cite{conforti1984submodular} in \thref{superioritytheorem}.
    \item Assuming that computations past the horizon are possible, we show that our bound is superior to the $\alpha_G''$ bound. The result is also part of \thref{superioritytheorem}.
    \item We provide a counterexample to show that the $\alpha_G'$ bound does not actually hold under the claimed conditions in \cite{conforti1984submodular}.
    \item We demonstrate the superiority of our bound to those in \cite{conforti1984submodular} with applications to monotone submodular sensor coverage problems with both set and string domains. 
    \item We apply our result to an social welfare maximization problem that is neither submodular nor subadditive and obtain strong performance bounds. 
\end{enumerate}

\section{Previous Work}
 Most of the previous work on performance guarantees of the greedy scheme has focused on the case where functions are \emph{monotone submodular} set functions, and can be placed into one of two categories. In what follows below, $\beta$ is used to represent the lower bound on the performance ratio of the greedy scheme as in the previous section.

 \begin{cat}
 \textbf{Class Bounds}
 \end{cat} \par
In the beautiful results by Fisher et al.\  \cite{fisher1978analysis} and Nemhauser et al.\ \cite{nemhauser1978analysis}, it is proven that that $\beta = 1/2$ over a finite general set matroid \cite{fisher1978analysis} and $\beta = 1- ((K-1)/K)^{K} > (1-e^{-1})$ over a finite uniform set matroid with $K$ being the horizon length \cite{nemhauser1978analysis}. The $(1-e^{-1})$ result was then extended to arbitrary matroids in \cite{calinescu2011maximizing}. The same $(1-e^{-1})$ bound has been derived for some string optimization problems in \cite{streeter2008online} and \cite{alaei2021maximizing}. Such results are strong in that they apply to the entire class of all monotone submodular functions over all matroids satisfying certain conditions. In other words, for whichever representative problem we choose from this class, the performance of the greedy scheme will be no worse than their aforementioned bounds. One weakness of such bounds, however, is that they do  not distinguish cases where employing the greedy scheme performs far better than their performance bound, which can happen in many instances. 

 \begin{cat}
   \textbf{Computationally Infeasible/Post-Horizon Bounds}
 \end{cat} \par
Other approaches to giving performance bounds for greedy schemes involve a sort of ``wishful thinking" wherein it is assumed that we can measure how much the outputs of our function decrease for the function $f$ for all inputs. Constants referred to as \textit{curvatures} are then constructed using these computations, and different kinds of curvature-based bounds are established as seen in \cite{liu2019improved}, \cite{zhang2015string}, and \cite{wang2016approximation}. Unfortunately, the amount of computation required to compute such curvatures for large action/state spaces is often as complex as the original optimization problem.\par
Other bounds happen to be computable in such a way that is easily extendable to large problems, and yield some strong results such as that of \cite{welikala2022new}. Unfortunately, these results involve computing values of the function on sets whose size exceeds the horizon length, which in many applications is not possible.\\
\textbf{The Conforti \& Cornu\'{e}jols  Bounds}

 In \cite{conforti1984submodular}, a set of three greedy curvature bounds proposed involving the constants $\alpha_G$, $\alpha_G'$, and $\alpha_G''$  happens to resolve the previous problems because they are computable, they yield bounding results that are specific to the problem in question, and $\alpha_G$ does not require computations past the horizon. Our present work builds on the $\alpha_G$ and $\alpha_G''$ bounds, and we will show later that they are inferior to our bound.

\noindent\textbf{Organization:} We begin with an introduction to strings and operations on them to set the stage in Section \rom{3}. We then explain how to translate the greedy curvature bounds from \cite{conforti1984submodular} to our notation in Section \rom{4}, and in Section \rom{5} we state and prove all of our main mathematical results. In Section \rom{6} we show applications of our result by computing performance bounds for a sensor coverage problem and a social welfare maximization problem. We conclude in section \rom{7} with a summary and some ideas for further research.

\section{Preliminaries for String Optimization}

In this section, we introduce the notation and definitions we will need to formulate the general string optimization problem. Let $f$ be our objective function with domain \gls{domain}. We seek to solve the following optimization problem:
\begin{equation}
\label{string_opt}
\begin{aligned}
    & \text{maximize } f(S) \\
    & \text{subject to } |S| = K, \quad S \in \bbT.
\end{aligned}
\end{equation} 

\begin{remark}
    We will refer to \gls{horizon} as the \textit{horizon length} for our decision problem. We obtain a string $S \in \bbT$ by starting with an empty string, and, at each stage, choosing a single symbol to add to our string at each of the $K$ stages. We give the necessary background for this process below.
\end{remark}

\begin{enumerate}
    \item Let \gls{groundset} be the \textit{ground set}.
    \item Elements $s \in \mathbb{S}$ are referred to as symbols.
    \item Let $s_1,s_2,\ldots,s_k \in \mathbb{S}$. Then the ordered $k$-tuple $S= (s_1, s_{2}, \ldots, s_{k})$ is a \textit{string} with length $|S| = k$.
    \item Let $\varnothing$ be the empty string, i.e., the string that contains no symbols.    
    \item For a string $S = (s_1, \ldots s_k)$, we use $S_{i} = (s_1, \ldots, s_i) \text{ for } i \in \{1,\ldots,k\}$ to represent the first $i$ symbols of $S$.
    \item Let $\mathbb{S}^*$ be the set of all strings of arbitrary length whose symbols are from $\mathbb{S}$.
    %\item The \textit{uniform string matroid of rank} $K$, denoted by $\mathbb{S}_{K}$, is the set of all strings $|S| \leq K$.
\end{enumerate}

\begin{remark}
Throughout the paper, we will avoid denoting strings using parentheses and instead write $(s_1, s_{2}, \ldots, s_{k})$ as $s_1s_{2} \cdots s_{k}$. Based on our definition of strings, we can see that two strings $S \text{ and } T$ are only equal when their lengths are equal, and $s_i = t_i$ for all $i \in \{1, \ldots, k\}$ where $k = |S| = |T|$. Thus, by the definitions above, we see that permutations of the symbols of a string might produce a distinct string from the original. 
\end{remark}

We now introduce an operation that can be performed on strings.
\begin{definition}
    Let  $S = s_{1}s_{2}\cdots s_{m}$ and $T = t_{1}t_{2}\cdots t_{n}$ belong to $\mathbb{S}^*$. The \textit{concatenation} of $S$ and $T$  is defined as the string $ST = s_{1}s_{2}\cdots s_{m}t_{1}t_{2}\cdots t_{n}$.
\end{definition}

The concatenation operation creates relationships between different strings that are central to the structure of $\bbT$, which necessitates the following definition. 
\begin{definition}
Let $P, S \in \mathbb{S}^*$. We say that $P$ is a \textit{prefix} of $S$, and write $P$\gls{curlyineq}$S$, if there exists a string $U \in \mathbb{S}^*$ such that $S = PU$.
\end{definition}

\begin{definition}
For any $S \in \bbT$, let \gls{allowedelts}$= \{s\in \mathbb{S}: Ss \in \bbT\}$.
\end{definition}

\begin{remark}
 The set $\bbS(S)$ contains the elements that are feasible to concatenate to a string $S$ and still remain in $\bbT$. The definition above will be used when referencing $\bbS(G_{k-1}) = \{s \in \bbS: G_{k-1}s \in \bbT\}$ and $\bbS(\varnothing) = \{s \in \bbS: s \in \bbT\}$. 
\end{remark}

\begin{definition}
    Any solution to the optimization problem (\ref{string_opt}) is said to be an \textit{optimal} string, and is denoted as \gls{optimal} $= o_{1}o_{2}\cdots o_{L}$. 
\end{definition}

\begin{definition}
    We define \gls{greedy} $= g_{1}g_{2} \cdots g_{K}$ to be a \textit{greedy solution} if for all $k \in \{1,2,\ldots,K\}$, 
    \[
    g_{k} = \arg\max_{s \in \mathbb{S}(G_{k-1})} f(g_{1}\cdots g_{k-1}s). 
    \]
\end{definition}

\begin{remark}
    Here, for simplicity of our arguments below we will assume that $L = K$ above, i.e., the greedy and optimal strings have the same length. Our formulation of the string optimization problem later on will allow for applications to problems wherein the optimal string has length smaller than that of the greedy string, and all of our arguments below still hold in this case. 
\end{remark}

% \addtolength{\textheight}{-3cm}   % This command serves to balance the column lengths

%%%%%%%%%%%%%%%%%%%%%%%%%%%%%%%%%%%%%%%%%%%%%%%%%%%%%%%%%%%%%%%%%%%%%%%%%%%%%%%%

\section{Notation}
Here, we explain how we translate results from the set-based notation in \cite{conforti1984submodular} to our string notation. In  \cite{conforti1984submodular}, a central definition is the \textit{discrete derivative} $\varrho_j(A) = f(A \cup \{j\}) - f(A)$ of $f$ along a set $A$, which measures how much a function changes after adding in an element $j$ to $A$. We now introduce our equivalent notation for strings.

\begin{definition}
    For $S \in \bbT$, and $u \in \mathbb{S}$, we define $\Delta(Su)= f(Su) - f(S)$, and refer to it as the \textit{increment of $f$ at the end of $Su$}. For the first $i$ elements $S_i$ of a string $S$, we define \gls{increment}$= f(S_i)-f(S_{i-1})$ and refer to this as the \textit{$i$th increment}.
\end{definition}

We now give a brief explanation by way of example as to how we rewrite the constants $\alpha_G \text{ and } \alpha_G''$ in \cite{conforti1984submodular} using our notation. Doing so for only $\alpha_G$ is sufficient as a very similar description produces the form for $\alpha_G''$.

We first rewrite $\alpha_G$ from \cite{conforti1984submodular} using our notation above:
\[
\alpha_G = \max_{2 \leq k \leq K} \max_{ s \in \mathbb{S}(G_{k-1}), f(s) > 0} 
    \left \{\frac{f(s)- \Delta(G_{k-1}s)}{f(s)} \right \}.
\]
Notice that the expression inside the curly braces can be written as $1-\Delta(G_{k-1}s)/f(s)$, and maximizing this expression is equivalent to minimizing $\Delta(G_{k-1}s)/f(s)$, which makes $1/(1-\alpha_G)$ equivalent to maximizing $f(s)/\Delta(G_{k-1}s)$. We then define:
\begin{equation}
\begin{aligned}
\mbox{\gls{gamma_G}} & = \frac{1}{1-\alpha_G}\\
& = \max_{2 \leq k \leq K} \max_{ s \in \mathbb{S}(G_{k-1}), \Delta(G_{k-1}s) > 0} \left \{ \frac{f(s)}{\Delta(G_{k-1}s)} \right \}.
\end{aligned}
\end{equation}

\begin{remark}
    Based on the definition used in \cite{conforti1984submodular}, the greedy algorithm continues as long as $\Delta(G_{k-1}g_k) \geq 0$. Because the value of $\gamma_G$ becomes undefined in the instances when $\Delta(G_{k-1}g_k) = 0$, we make assumption $\mathbf{A}_5$ later for certain proofs in the main results section. We note that imposing condition $\mathbf{A}_5$ is done solely for the purpose of generalizing the curvature constants of \cite{conforti1984submodular}, and is not necessary for our performance bound to hold.
\end{remark}
\iffalse
\begin{remark}
    Based on the definition used in \cite{conforti1984submodular}, the greedy algorithm continues as long as $\Delta(G_{k-1}g_k) \geq 0$.
\end{remark}
\fi

The definition of $\alpha_G$ in \cite{conforti1984submodular} generalizes easily to strings, but the same is not true for $\alpha_G''$. Implicit in the definition of $\alpha_G''$ is the assumption that for every element $s$ of the ground set of our set matroid, the function can be evaluated on the entire ground set of the matroid minus the element $s$. For set matroids, the order of the elements in the ground set does not influence the value of $f$. Because permutation invariance of the arguments is not a property shared between functions on strings and functions on sets, we need to introduce definitions to ensure that the string version of $\alpha_G''$ reduces to the definition in \cite{conforti1984submodular} for set matroids. First, we describe strings with length $K-1$ where $K$ is the horizon length, and reference the element $s$ that we are appending to the end of the string. We do so because such strings are used in the definition of $\alpha_G''$ in \cite{conforti1984submodular}.

\begin{definition}
    Let \gls{sigma} be the set of strings $S$ of length $K-1$ comprised of symbols in $\bbS$ such that $Ss \in \bbT$ . 
\end{definition}

\begin{remark}
    The restriction of ``strings $S$ of length $K-1$ comprised of symbols in $\bbS$ such that $Ss \in \bbT$" is made to exclude the strings $S$ composed of elements from $\bbS$ such that $Ss \not \in \bbT$. We see that $\gamma_G''$ generalizes $\alpha_G''$ when $\bbT$ is a string matroid and $K = |\bbS|$.
\end{remark}

Using the same algebraic manipulations used to obtain the formula for $\gamma_G$, we can obtain the formula for $\gamma_G''$ seen below. Because we seek to maximize $\gamma_G''$ and generalize $\alpha_G''$, we would like the denominator of the expression in \eqref{eq:gammagdp} below to be as small as possible. Therefore, for a specific $s \in \bbS$ we would like to choose the string that produces the smallest increment increase under $f$. Our discussion above motivates the following definition. 
\begin{definition}
    Let $s \in \bbS$. Then define \gls{deltastar}$= \min_{S \in \sigma_{\bbT}(s, K)} (f(Ss) - f(S))$.
\end{definition}

\indent We would like our generalized version of the $\alpha_G''$ bound to be well defined for strings, and to produce as good of a performance bound as possible. In our notation, we rewrite the $\alpha_G''$ from \cite{conforti1984submodular} as \[ \alpha_G'' = \max_{s\in \mathbb{S}\setminus{\{g_1\}}} \left\{ 1- \frac{\Delta^*(s)}{f(s)}\right\} .\] For the same reason as for $\alpha_G$, we define 

\begin{equation}
\label{eq:gammagdp}
\mbox{\gls{gamma_G''}}= \frac{1}{1-\alpha_G''} = \max_{s\in \mathbb{S}\setminus{\{g_1\}}} \left \{\frac{f(s)}{\Delta^*(s)} \right \}
\end{equation}

\begin{remark}
We see that in the permutation-invariant case of set matroids where $K = |\bbS|$, the definition of $\Delta^*(s)$ reduces to the quantity $f((\bbS \setminus s)s) - f(\bbS \setminus s)$. Our definition of $\gamma_G^{''}$ then agrees with the expression $1/ (1-\alpha_G^{''})$ in \cite{conforti1984submodular}.
\end{remark}

\begin{remark}
The reader may recognize that for strings, the generalization of $\alpha_G''$ to $\gamma_G''$ is not computable for large $\bbS$. We will show later that this is inconsequential since our computable bound still beats this generalization.
\end{remark}

\section{A General Performance Bound}

\subsection{Main Results}
 We established in \cite{van2023improved} that \emph{all} problems to which the $\alpha_G$ bound of \cite{conforti1984submodular} applies are special cases of the string results presented below. In \cite{li2024bounds}, we simplified the results from \cite{van2023improved} and slightly generalized them. In this section, we prove results that are more general than both \cite{van2023improved} and \cite{li2024bounds} in that they apply to a wider class of functions with less restrictions on their domains. 
 
 We first introduce assumptions $\mathbf{A_1}$ through $\mathbf{A_7}$, which are used throughout the section. Afterward, we define the mathematical objects that our assumptions are supposed to generalize, namely, monotone set and string submodularity for functions, as well as finite rank set and string matroids. Then, we generalize the greedy curvature bounds based on $\gamma_{G}$ and $\gamma_{G}''$ in \cite{conforti1984submodular} to the string setting in ~\thref{GenConforti}. We then establish an easily computable bound relying on fewer assumptions in~\thref{lemmatopk}. Finally, ~\thref{superioritytheorem} shows that our bound is provably better than those generalized greedy curvature bounds in~\thref{GenConforti}.

\noindent\textbf{Assumptions:} The assumptions listed here are used in the proofs of all the results below. \\
$\mathbf{A_1}$: ($G_K$-feasiblity) For each $i \in \{1,\ldots,K\}$, $o_{i} \in \mathbb{S}(G_{i-1}) \cap \bbS(\varnothing)$. \\
$\mathbf{A_2}$: ($O_K$-subadditivity) $f(O_K) \leq \sum_{i=1}^K f(o_i)$.\\
$\mathbf{A_3}$: (Null on null)   $f(\varnothing) =0$ (recall $\varnothing$ is the empty string).\\
$\mathbf{A_4}$: (Nondecreasing) For all $P, S \in \bbT$ such that $P\preccurlyeq S$, $f(P) \leq f(S)$.\\
$\mathbf{A_5}$: ($\Delta$-positivity) For all $s \in \bbS(G_i)$ and $i \in \{1, \ldots, K-1\}$ we have $\Delta(G_is) > 0$.\\
$\mathbf{A_6}$: ($\Delta^*(u)$-Diminishing returns) For all $i \in \{1, \ldots, K\} \text{ and all } u \in \bbS(G_i)$, we have that $\Delta(G_iu) \geq \Delta^*(u)$.\\
$\mathbf{A_7}$: For all $u \in \bbS$,  $\Delta^*(u) >0$.

\begin{remark}
Given that we can assume without loss of generality that $f(\varnothing) = 0$, we note that assumption $\mathbf{A}_3$ is only included above because of its frequent use in the proofs that generalize the curvature constants of \cite{conforti1984submodular} to strings. We also include $\mathbf{A}_7$ solely for the purpose of ensuring that $\gamma_G''$ is well defined. 
\end{remark}

We now define mathematical objects whose properties are widely used in existing literatures on submodular optimization, and show that our assumptions listed above are generazlizations of these objects. 

\begin{definition}
Let $\bbS$ be any ground set and $\bbX$ a family of subsets of $\bbS$. We say that
$(\bbS, \bbX)$ is a \textit{set matroid of finite rank $K$} if 
\begin{enumerate}
    \item For all $S \in \bbX$, $|S| \leq K$.
    \item For all $S \in \bbX$, $T \subset S$ implies $T \in \bbX$.
    \item For all $S, T \in \bbX$ where $|T| +1 = |S|$, there exists $j \in S \setminus T$ such that $T \cup \{j\} \in \bbX$.
\end{enumerate}
\end{definition}

\begin{remark}
    Note that part $2$ of the above definition implies that $\mathbf{A}_1$ is a generalization of \cite{conforti1984submodular}, since $\{o_i\} \subset \{o_1, \ldots, o_K\}$ implies that each $\{o_i\} \in \bbS(\varnothing)$. The fact that $o_i \in \bbS(G_{i-1})$ follows from Lemma 2.2 of \cite{conforti1984submodular}. 
\end{remark}

\begin{definition}
We say that a function $f: \bbX \to \mathbb{R}$ on a finite rank set matroid $\bbX$ is \textit{monotone set submodular} if 

\begin{enumerate}
    \item $f$ is \textit{monotone}, i.e., for all $A \subseteq B \in \bbX, f(A) \leq f(B)$.
    \item $f$ has the \textit{diminishing returns} property, i.e., for all $ A, B \in \bbX$ such that $A \subseteq B$, and for all $a \in \bbS \text{ that are feasible at } A \text{ and } B, f(A \cup a)-f(A) \geq f(B \cup a) - f(B)$. 
\end{enumerate}
\end{definition}

\begin{remark}
Monotonicity is captured by $\mathbf{A}_4$. The diminishing-returns property defined above, along with assumption $\mathbf{A}_3$, tells us that $f(o_{j+1}) \geq f(O_{j+1}) - f(O_j)$. If we compute the sum of the terms on the left hand side of the inequality and the sum of the terms on the right hand side of the inequality for $j \in \{0, \ldots, K-1\}$, we obtain $\mathbf{A}_2$.
\end{remark}

\indent We now introduce the string generalizations of the above definitions. 

\begin{definition}
Let $A \in \bbS^*$ be a string. Then the \textit{components} of $A$ is the subset $C(A) \subset \bbS$ where $a \in C(A)$ if $a$ appears at least once in the string $A$.
\end{definition}

\begin{definition}
Let $\bbS$ be our ground set, and $\bbT \subset \bbS^*$. Then $\bbT$ is a \textit{string matroid of finite rank $K$} if

\begin{enumerate}
    \item For all $A \in \bbT$, $|A| \leq K$.
    \item For all $B \in \bbT$, $A \preccurlyeq B$ implies $A \in \bbT$.
    \item For all $A, B \in \bbT$ where $|A|+1 = |B|$, there exists $a \in C(B)$ such that $A a \in \bbT$.
\end{enumerate}
\end{definition}

\begin{definition}
We say that a function $f: \bbT \to \mathbb{R}$ on a finite rank string matroid $\bbT$ is \textit{monotone string submodular} if 

\begin{enumerate}
    \item $f$ has the \textit{forward monotone} property, i.e., for all $A \preccurlyeq B \in \bbT, f(A) \leq f(B)$.
    \item $f$ has the \textit{diminishing returns} property, i.e., for all $A \preccurlyeq B \in \bbT$, for all $a \in \mathbb{S} \text{ such that } Aa \text{ and } Ba \in \bbT, f(Aa)-f(A) \geq f(Ba) - f(B)$. 
\end{enumerate}
\end{definition}

\begin{remark}
    The generalized greedy curvature bound based on $\gamma_{G}$ in \cite{conforti1984submodular} relies on assumptions $\mathbf{A_1}, \ldots, \mathbf{A_5}$. The other generalized greedy curvature bound (based on $\gamma_{G}''$) in \cite{conforti1984submodular} relies on all the assumptions from $\mathbf{A_1}$ to $\mathbf{A_7}$. The additional $\mathbf{A_6}$ and $\mathbf{A_7}$ are needed to account for the fact that $\gamma_G''$ requires computations beyond the horizon. However, our bound, relying on only $\mathbf{A_1}$ and $\mathbf{A_2}$, is provably better than both of the generalized greedy curvature bounds.
\end{remark}

~\\
\indent We need the following proposition and lemma for the proofs of later theorems. 

\begin{proposition}\thlabel{prop}
The inequality \[\sum_{i=1}^Kf(o_i) \leq f(g_1) + \gamma \sum_{i=2}^K\Delta(G_i)\] holds when:
\begin{enumerate}
\item $\gamma = \gamma_G$ assuming $\mathbf{A_1, A_3, A_4}$, and $\mathbf{A_5}$.
\item $\gamma = \gamma_G''$ assuming $\mathbf{A_1, A_3, A_4, A_5, A_6,}$ and $\mathbf{A_7}$.
\end{enumerate}

\end{proposition}

\begin{proof}
By assumptions $\mathbf{A_3}$ and $\mathbf{A_4}$, all quantities involved are nonnegative. By definition of the greedy string, we have $f(g_1) \geq f(o_1)$. 

\noindent 1) The definition of $\gamma_{G}$ and assumption $\mathbf{A_5}$ give us:
    \begin{equation}
    \begin{aligned}
        & \text{For all } \; i\in\{2,\ldots,K\} \text{ and } s\in\mathbb{S}(G_{i-1}), \\ 
        & f(s) \leq \gamma_{G} \Delta(G_{i-1}s).
    \end{aligned}
    \end{equation}
Then by assumption $\mathbf{A_1}$ and the definition of greedy string, we have:
\begin{equation}
    f(o_{i}) \leq \gamma_{G} \Delta(G_{i-1}o_{i}) \leq \gamma_{G}\Delta(G_{i}) \text{ for } i \in \{2,\ldots,K\}.
\end{equation}
Therefore, 
\begin{equation}
    \sum_{i=1}^{K}f(o_i) \leq f(g_1) + \gamma_{G}\sum_{i=2}^{K} \Delta(G_{i}).
\end{equation}
    
\noindent 2) Using the definition of $\gamma_{G}''$ and the additional assumptions $\mathbf{A_6}$ and $\mathbf{A_7}$, we have:
\begin{equation}
\begin{aligned}
    & f(o_i) \leq \gamma_{G}''\Delta^*(o_i) \leq \gamma_{G}''\Delta(G_{i-1}o_{i}) \leq \gamma_{G}''\Delta(G_{i}), \\
    & \text{for } i \in \{2,\ldots,K\}. 
\end{aligned}
\end{equation}
Therefore, 
\begin{equation}
    \sum_{i=1}^{K}f(o_i) \leq f(g_1) + \gamma_{G}''\sum_{i=2}^{K} \Delta(G_{i}).
\end{equation}
%\qed
\end{proof}

% \begin{remark}
% Assumption $\mathbf{A_2}$ combined with \thref{prop} implies that $\gamma \geq 1$ always holds whenever $\gamma = \gamma_{G} \text{ or } \gamma_{G}''$. If $\gamma < 1$, then 
% \begin{equation*}
% \begin{aligned}
%     f(O_K) & \leq \sum_{i=1}^{K}f(o_i) \leq f(g_1) + \gamma \sum_{i=2}^K\Delta(G_i) \\
%     & < f(g_1) + \sum_{i=2}^K\Delta(G_i) = f(G_K),
% \end{aligned}
% \end{equation*}
% which gives us contradiction with $f(O_K) < f(G_K)$.
% \end{remark}

\begin{lemma}\thlabel{Blemma}
     If an upper bound $B$ of $f(O_K)$ satisfies 
    \begin{equation*}
        B \leq f(g_1) + \gamma \sum_{i=2}^{K} \Delta(G_i), 
    \end{equation*}
    
\noindent we have \[ \frac{f(G_K)}{f(O_K)} \geq \frac{f(G_K)}{B} \geq \frac{1}{\gamma} + \left(1 - \frac{1}{\gamma} \right)\frac{f(g_1)}{B}, \]
where $\gamma = \gamma_{G} \text{ or } \gamma_{G}''.$
\end{lemma}

\begin{proof}
The inequality $B \leq f(g_1) + \gamma \sum_{i=2}^{K} \Delta(G_i)$ can be rewritten as:

\begin{equation}
\label{zerotrick}
\begin{aligned}
    B & \leq (1-\gamma) f(g_1) + \gamma f(g_1) + \gamma \sum_{i=2}^{K}\Delta(G_i) \\
    & = (1-\gamma)f(g_1) + \gamma f(G_K).
\end{aligned}
\end{equation}

Considering $f(O_K) \leq B$ and dividing $\gamma B$ on both sides of \eqref{zerotrick}, we have 

\begin{equation}
\begin{aligned}
      \frac{f(G_K)}{f(O_K)} \geq \frac{f(G_K)}{B} \geq \frac{1}{\gamma} + \left(1-\frac{1}{\gamma} \right) \frac{f(g_1)}{B}.
\end{aligned}
\end{equation}
%\qed
\end{proof}

% \begin{lemma}\thlabel{Blemma}
%     Given an upper bound $B$ of $f(O_K)$ which satisfies:
%     \[\sum_{i=1}^{K}f(o_i) \leq B \leq f(g_1) + \gamma\sum_{i=2}^K\Delta(G_i) ,\]

% \noindent we have \[ \frac{f(G_K)}{f(O_K)} \geq \frac{f(G_K)}{B} \geq \frac{1}{\gamma} + \left(1 - \frac{1}{\gamma} \right)\frac{f(g_1)}{B}, \]
% where $\gamma = \gamma_{G} \text{ or } \gamma_{G}''.$
% \end{lemma}

% \begin{proof}
% Assumptions $\mathbf{A_1, A_3, A_4, A_5, A_6}$ and $\mathbf{A_7}$ guarantee that \thref{prop} holds, and assuming $\mathbf{A_2}$ we obtain the inequality \[f(O_K) \leq B \leq f(g_1) + \gamma\sum_{i=2}^K\Delta(G_i).\] We also note that by adding $0 = \gamma f(g_1) -\gamma f(g_1)$ to the right hand side, we can rewrite it as $f(g_1) + \gamma\sum_{i=2}^K\Delta(G_i) = (1-\gamma)f(g_1) + \gamma f(G_K)$.
%     By assumptions $\mathbf{A_3}$ and $\mathbf{A_4}$, we see that all values of the objective function are nonnegative, and in particular $0 < f(O_K) \leq B$. Thus, dividing both sides by $B$ and solving for $f(G_K)/B$ we obtain:
%     \begin{equation}
%         \begin{aligned}
%              1  & \leq (1-\gamma)\frac{f(g_1)}{B} + \gamma \frac{f(G_K)}{B}\\
%              \frac{1}{\gamma} + \left(1 - \frac{1}{\gamma} \right)\frac{f(g_1)}{B}& \leq \frac{f(G_K)}{B} \stackrel{(a)}{\leq} \frac{f(G_K)}{f(O_K)}.
%         \end{aligned}
%     \end{equation}
%     where inequality $(a)$ follows from the fact that $B \geq f(O_K)$.
% \end{proof}

~\\
\indent We now generalize the greedy curvature bounds in \cite{conforti1984submodular} to strings.

\begin{theorem}\thlabel{GenConforti}
Assuming that $\mathbf{A_1, \ldots, A_7}$ hold, the greedy curvature bounds in \cite{conforti1984submodular} generalize to the string setting in our notation as 
\[
\frac{f(G_K)}{f(O_K)} \geq \frac{1}{K} + \frac{1}{\gamma}\frac{K-1}{K},
\]
where $\gamma = \gamma_{G} \text{ or } \gamma_{G}''$.

\end{theorem}

\begin{proof}
Assumption $\mathbf{A}_1$ and the definition of the greedy string imply that $f(g_1) \geq f(o_i)$ for all $i \in \{1, \ldots, K\}$. Combining this fact with assumption $\mathbf{A}_2$ we obtain: 
\begin{equation}
    f(O_K) \leq \sum_{i=1}^{K}f(o_i) \leq Kf(g_1). 
\end{equation}

Let $B = \sum_{i=1}^{K}f(o_i)$. Then, by \thref{prop} and \thref{Blemma}, we know: 
\begin{equation}
        \begin{aligned}
             \frac{f(G_K)}{f(O_K)} &\geq \frac{f(G_K)}{B} \geq \frac{1}{\gamma} + \left(1 - \frac{1}{\gamma} \right)\frac{f(g_1)}{B} \\
             & \stackrel{(a)}{\geq} \frac{1}{\gamma} + \left(1 - \frac{1}{\gamma} \right)\frac{f(g_1)}{Kf(g_1)} \\
             & = \frac{1}{K} + \frac{1}{\gamma}\frac{K-1}{K}, 
        \end{aligned}
    \end{equation}
where inequality (a) is because $\gamma > 1$ and $B \leq Kf(g_1)$.
%\qed
\end{proof}

~\\
\indent We define our upper bound for $f(O_K)$ as
\[
B_s = \sum_{i=1}^Kc_i, \text{ where } c_i = \max_{s \in \mathbb{S}(G_{i-1})\cap \bbS(\varnothing)} f(s).
\]

We now prove that $B_s$ is an upper bound for $f(O_K)$.
\begin{theorem}
\thlabel{lemmatopk}
Assuming $\mathbf{A_1}$ and $\mathbf{A_2}$, we have that $f(O_K) \leq B_s$.
\end{theorem}

\begin{proof}
\begin{equation}
    \begin{aligned}
        f(O_{K}) \stackrel{(a)}{\leq} \sum_{i=1}^{K}f(o_{i}) \stackrel{(b)}{\leq} \sum_{i=1}^{K}  c_{i} = B_{s}.
    \end{aligned}
\end{equation}
Inequality $(a)$ follows from $\mathbf{A_{2}}$ while inequality $(b)$ holds because of $\mathbf{A_1}$ and the definition of $c_{i}$.
%\qed
\end{proof}

\begin{remark}
\thlabel{applicationuse}
Given some additional information about the objective function, it is possible that we do not need the requirement that $o_i \in \bbS(G_{i-1})$ in assumption $\mathbf{A_1}$ for $B_s$ to be an upper bound for $f(O_K)$. For example, if we know that all the elements of the ground set are feasible at epoch one and each element can be chosen only once throughout all the epochs, then we can take $B_s$ to be the sum of the $K$ largest objective-function values at epoch one. This is the case for one of the string problems in Section \rom{6}. 
\end{remark}

~\\
\indent The following theorem shows that our bound $B_{s}$, relying on fewer assumptions, is better than both the $\gamma_G$ and $\gamma_G''$ greedy curvature bounds in \cite{conforti1984submodular}. 

\begin{theorem}\thlabel{superioritytheorem}
Assuming $\mathbf{A_1, \ldots, A_7}$, the $B_s$ performance bound is superior to the generalized greedy curvature bounds in \cite{conforti1984submodular}, i.e., \[\frac{f(G_K)}{f(O_K)} \geq \frac{f(G_K)}{B_s} \geq \frac{1}{K} + \frac{1}{\gamma}\frac{K-1}{K}, \]
where $\gamma = \gamma_{G} \text{ or } \gamma_{G}''$.
\end{theorem}

\begin{proof}
    By the definition of $g_{1}$, $B_s = \sum_{i=1}^{K} c_{i} \leq Kf(g_1)$. Let $s_{i}$ denote the element that achieves the maximum $c_{i}$ for each $i \in \{1,\ldots,K\}$. The definition of $c_{1}$ gives us $s_{1} = g_{1}$. Then, by assumptions $\mathbf{A_5}, \mathbf{A_6}, \mathbf{A_7}$ and the definitions of $\gamma_{G}$ and $\gamma_{G}''$, we have
    \begin{equation}
    \label{suffice}
    \begin{aligned}
        B_{s} & = f(g_{1}) + \sum_{i=2}^{K} f(s_{i}) 
        \leq f(g_{1}) + \gamma\sum_{i=2}^{K}\Delta(G_{i-1}s_{i}) \\
        & \leq f(g_{1}) + \gamma \sum_{i=2}^{K}\Delta(G_{i}),
    \end{aligned}
    \end{equation}
    where $\gamma = \gamma_{G} \text{ or } \gamma_{G}''$.
    
    Finally, by \thref{Blemma}, \thref{GenConforti}, and \thref{lemmatopk}, we obtain
    \begin{equation}
    \begin{aligned}
        \frac{f(G_K)}{f(O_K)} & \geq \frac{f(G_K)}{B_s}  \geq \frac{1}{\gamma} + \left(1-\frac{1}{\gamma} \right) \frac{f(g_1)}{B_s} \\
        & \geq \frac{1}{\gamma} + \left(1-\frac{1}{\gamma} \right) \frac{f(g_1)}{Kf(g_1)} \\
        & = \frac{1}{K} + \frac{1}{\gamma}\frac{K-1}{K}.
    \end{aligned}
    \end{equation}
    %\qed
\end{proof}

% \begin{proof}
%      Note that by definition of $g_1$, $f(g_1) \geq c_i$ for $i \in \{1, \ldots, K\}$. Therefore, $Kf(g_1) \geq B_s$. Assumptions $\mathbf{A_1, \ldots, A_7}$  enable us to conclude that by \thref{Blemma} and \thref{lemmatopk} it suffices to show that $B_s \leq f(g_1) + \gamma\sum_{i=2}^K\Delta(G_i)$. By definition, $s_1 = g_1$, so it remains to show that the $i =2, \ldots, K$ terms satisfy $f(s_i) \leq \gamma \Delta(G_i)$ or equivalently by assumption $\mathbf{A_5}$, $f(s_i)/\Delta(G_i) \leq \gamma$. When $\gamma = \gamma_G$, we obtain
%     $f(s_i)/\Delta(G_i) \leq  f(s_i)/\Delta(G_{i-1}s_i) \leq \gamma_G$ by definition of $\gamma_G$ and the greedy string. The same argument works for $\gamma_G''$, with the exception that to obtain the last inequality, we use assumptions $\mathbf{A_6}$ and $\mathbf{A_7}$ to obtain $f(s_i)/\Delta(G_{i-1}s_i) \leq f(s_i)/\Delta(M_{\bbT}(\bbS \setminus s_i)s_i)) \leq \gamma_G''$ by definition of $\gamma_G''$.

%     We can now apply \thref{Blemma} to obtain \begin{equation}
%         \begin{aligned}
%              \frac{f(G_K)}{f(O_K)}  \geq \frac{f(G_K)}{B_s}
%          & \geq \frac{1}{\gamma} + \left(1 - \frac{1}{\gamma} \right)\frac{f(g_1)}{B_s} \\
%          & \stackrel{(a)}{\geq} \frac{1}{\gamma} + \left(1 - \frac{1}{\gamma} \right)\frac{f(g_1)}{Kf(g_1)}\\
%          & \stackrel{(b)}{\geq} \frac{1}{K} + \frac{1}{\gamma}\cdot \frac{K-1}{K}.
%         \end{aligned}
%     \end{equation}
%     where inequality $(a)$ follows from the fact that $B_s \leq Kf(g_1)$ and inequality $(b)$ follows from simple algebra.
% \end{proof}

\subsection{A Counterexample}
In this section, we focus on the case of monotone submodular set functions with matroid domains. We keep our notation the same but make the following two assumptions as doing so reduces us to the set case:

\begin{enumerate}
    \item Our function $f$ does not change value if the symbols of a string are permuted.
    \item In each string in $\bbT$, each symbol can be used at most once.
\end{enumerate} 
After proving their greedy curvature bound for $\alpha_G$, Conforti and Cornu\'{e}jols \cite{conforti1984submodular} introduce two other constants for which they claim ``Theorem 3.1 remains true if $\alpha_G$ is replaced by the parameters \[ \alpha_G' = \max_{j \in \bbS \setminus{G_{K-1}}, f(j) > 0} \left\{ \frac{f(j)-\Delta(G_{K-1}j)}{f(j)}   \right\}\] or 

\[ \alpha_G'' = \max_{j \in \bbS \setminus{g_1}, f(j) > 0} \left\{ \frac{f(j)-\Delta((\bbS \setminus{j})j)}{f(j)}   \right\}."\]
While this claim is true for $\alpha_G''$, as we showed in \thref{GenConforti} above, we now show that the method of proof of the $\alpha_G'$ bound is incorrect. The crux of the argument in \cite{conforti1984submodular} is the following sequence of inequalities:

\[ f(O_K) \stackrel{(a)}{\leq} \sum_{i=1}^K f(o_i) \stackrel{(b)}{\leq} f(g_1) +  \frac{1}{1-\alpha_G'}\sum_{k=2}^K \Delta(G_{k}). \]

The key to our counterexample is getting inequality $(b)$ to fail for $\alpha_G'$, i.e., an example wherein:
\begin{equation} \label{eq1}
    \sum_{k=1}^Kf(o_k) > f(g_1) + \frac{1}{1-\alpha_G'}\sum_{k=2}^K \Delta(G_{k}).
\end{equation}

If this inequality holds, then their proof method fails. Our context is the following:
\begin{itemize}
    \item The ground set of the matroid is $\bbS = \{w, x, y, z\}$.
    \item Our matroid is the set $\bbT = \{S \in \bbS^*: |S| \leq 3\}$.
    \item The function $f: \bbT \to \mathbb{R}_{\geq 0}$ is monotone submodular and $f(\varnothing) = 0$.
    \item The horizon length is $K = 3$.
    \item In our example, $G_1 = O_1 = y$, $G_2 = O_2 = xy$, and $G_3 = O_3 = xyz$.
\end{itemize}

Before defining $f$ on all of $\bbT$, we first focus on the key parts of our counterexample wherein the definition of $f$ satisfies the above inequality. We assume that the optimal and the greedy solutions are both equal to $xyz$. We define \[ f(w) = \frac{1}{3}, \quad f(x) = 9, \quad f(y) = 10, \quad f(z) = \frac{1}{2}\] and \[ \Delta((xy)w) = \frac{1}{6}, \quad \Delta((xy)z) = \frac{1}{3}.\] Here, since $G_2 = xy$, we see that $\bbS \setminus{G_2} = \{w, z\}$. Now note that \begin{equation} \frac{1}{1-\alpha_G'} = \max_{j \in \bbS\setminus{G_2}} \left \{ \frac{f(j)}{\Delta((xy)j)}\right \}, \end{equation} which after computing the quantity inside of the brackets of equation $(6)$ for $j = z$ we get $3/2$, and for $j=w$ we obtain $2$. Therefore, plugging in our values into the inequality, we get: \begin{equation}
    \begin{split}
        \sum_{k=1}^Kf(o_k) & = 10 + 9 + \frac{1}{2} \\
        & > 10 + 2(1 + \frac{1}{3}) \\
        & = f(g_1) + \frac{1}{1-\alpha_G'}\sum_{k=2}^K \Delta(G_{k}).
    \end{split}
\end{equation}

We see from the computations in $(7)$ that the inequality from \thref{prop} does not hold. We now define $f$ on the power set of $\bbS$ in Table \ref{tab1}.

\begin{table}[H]
\centering
\caption{$f(S)$ for $S \in \bbT$}
\label{table}
\setlength{\tabcolsep}{3pt}
\begin{tabular}{|p{50pt}|p{50pt}|}
\hline
$|S| = 0$&
$f(S)$\\
\hline
$\varnothing$&
$0$\\
\hline
$|S| =1$&
$f(S)$\\
\hline
$w$&
$\frac{1}{3}$\\

$x$&
$9$\\

$y$&
$10$\\

$z$&
$\frac{1}{2}$\\

\hline
$|S| =2$&
$f(S)$\\
\hline
$wx$&
$9 + \frac{1}{3}$\\

$wy$&
$10 + \frac{1}{3}$\\

$wz$&
$\frac{1}{3} + \frac{1}{2}$\\

$xy$&
$9+2$\\

$xz$&
$9 + \frac{1}{2}$\\

$yz$&
$10 + \frac{1}{2}$\\

\hline
$|S| =3$&
$f(S)$\\
\hline
$wxy$&
$9 + \frac{1}{3} + \frac{11}{6}$\\

$wxz$&
$9 + \frac{1}{3} + \frac{1}{2}$\\

$wyz$&
$\frac{1}{3} + 10 + \frac{1}{2}$\\

$xyz$&
$9 + \frac{1}{2} + \frac{11}{6}$\\
\hline
$|S| = 4$&
$f(S)$ \\
\hline
$wxyz$&
$11 + \frac{1}{3}$\\
\hline
\end{tabular}
\label{tab1}
\end{table}

\noindent Checking that the function is monotone 
 submodular is not difficult, and with the example above we have established that the $\alpha_G'$ bound in \cite{conforti1984submodular} has an incorrect proof. We note that whether the bound itself holds or not is unclear, as the inequality  $11 + 1/3 = f(O_K) \leq f(g_1) + 1/(1-\alpha_G')\sum_{k=2}^K \Delta(G_{k}) = 12 + 2/3$ still  still holds for our particular example. The counterexample shows that the proof of this result is invalid, and calls into question the correctness of the theorem. Remark 3.4 of \cite{conforti1984submodular} claims that $\alpha_G'$ is also valid to use in the formula for the bound in Theorem 3.1 of \cite{conforti1984submodular}, and since $\alpha_G'$ is computable, a bound involving its computation is practically appealing. However, until we provided this counterexample, it was unknown that the theorem could be wrong.

%%%%%%%%%%%%%%%%%%%%%%%%%%%%%%%%%%%%%%%%%%%%%%%%%%%%%%%%%%%%%%%%%%%%%%%%%%%%%%%%

\section{APPLICATIONS}
Our applications in this section aim to demonstrate the superiority of our result to those of \cite{conforti1984submodular} in two different scenarios, and to apply our results in two more scenarios in which, to the best of our knowledge, is the first of its kind. We begin by showing the strength of our result in the classical case of monotone submodular set functions on matroids via a sensor coverage problem. We then modify the sensor coverage problem by adding a time-varying component that results in a string optimization problem with a monotone submodular string objective function. We conclude with a social welfare maximization problem with black-box utility functions that are neither submodular nor subadditive. We again add a time-varying component to the utility functions of the agents to produce a string optimization problem.

%%%%%%%%%%
\subsection{Multiagent Sensor Coverage}
The multiagent sensor coverage problem was originally studied in \cite{zhong2011distributed} and further analyzed in \cite{sun2019exploiting} and \cite{welikala2022new}. In a given mission space, we seek to find a placement of a set of sensors to maximize the performance of detecting randomly occurring events. We apply our results to a discrete version of this problem where the placement of the sensors is constrained to lattice points, and the occurrence of random events is possible at any point within the mission space. The settings in this application are more generalized than the same problem we studied in \cite{van2023improved} and \cite{li2024bounds}.

The mission space $\Omega \subset \mathbb{R}^{2}$ is modeled as a non-self-intersecting polygon where a total of $K$ sensors will be placed to detect a randomly occurring event. Those lattice points feasible for sensor placement within the mission space are denoted by $\Omega^{F} \subset \mathbb{R}^{2}$. The placement of these sensors occurs over $K$ epochs, wherein one sensor is placed during each epoch. Our goal is to maximize the overall detection performance in the mission space, as illustrated in Fig.~\ref{sensors1}. 

\begin{figure}[hbt!]
    \centering
    \includegraphics[width=0.98\columnwidth]{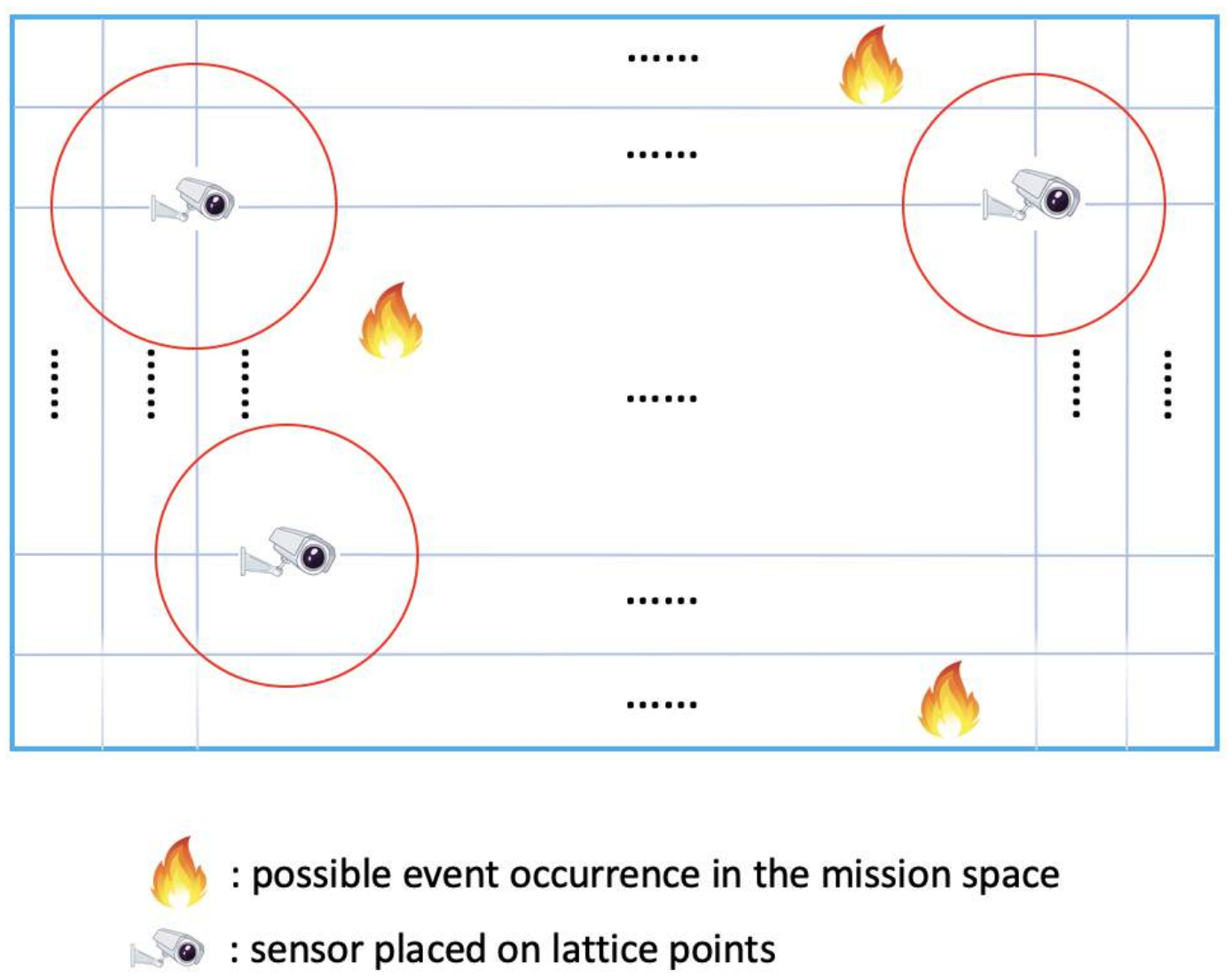}
    \caption{Sensor coverage for event detection in a mission space.} 
    \label{sensors1}
\end{figure}

 The location of a sensor placed during epoch $i$ is denoted by $\mathbf{s}_{i}$, and it can detect any occurring event at location $\mathbf{x} \in \Omega$ with probability $p(\mathbf{x},\mathbf{s}_{i}) = \exp{(-\lambda_{i} \|\mathbf{x} - \mathbf{s}_{i}\|)}$. We assume that $\lambda_{i} > 0$, where $\lambda_i$ is the coverage decay rate of that sensor at epoch $i$ for $i\in \{1, \ldots, K\}$. Once all $K$ sensors have been placed, we represent the choice of these sensors with the string $S = (\mathbf{s}_1,\mathbf{s}_2,\ldots, \mathbf{s}_K) \in \left(\Omega^{F}\right)^K$. The set and string case are distinguished by the time-varying properties of the sensors, which we will refer to as homogeneous and nonhomogeneous sensors, respectively.  \\

\noindent \textbf{Homogeneous Sensors:}
We assume that all the sensors are working independently, and their decay rates are all equal and remain fixed over epochs $1, \ldots, K$, i.e., $\lambda_{1} = \cdots = \lambda_{K} = \lambda$. Thus, we can determine that the probability of detecting a randomly occurring event is just one minus the product of the probabilities that each sensor fails to detect that event. More specifically, the probability of detecting an occurring event at location $\mathbf{x} \in \Omega$ after placing $K$ homogeneous sensors at locations $\mathbf{s}_i$ for $i \in \{1, \ldots, K\}$ is given by
\begin{equation}
    P(\mathbf{x}, S) = 1-\prod_{i=1}^{K}\left( 1-p(\mathbf{x},\mathbf{s}_i) \right) = 1-\prod_{i=1}^{K}\left( 1- e^{-\lambda \|\mathbf{x}-\mathbf{s}_i\|} \right). 
\end{equation}

To calculate the detection performance over the entire mission space $\Omega$, we incorporate the event density function $R$. The event occurrence over $\Omega$ is characterized by an event density function $R: \Omega \xrightarrow{} \mathbb{R}_{\geq 0}$, and we assume that $\int_{\mathbf{x} \in \Omega} R(\mathbf{x}) d\mathbf{x} < \infty$. Our objective function then becomes $H(S) = \int_{\mathbf{x} \in \Omega} R(\mathbf{x})P(\mathbf{x},S) d\mathbf{x}$, and we have the following string optimization problem:
\begin{equation}
\begin{aligned}
\label{obj_fun_sensor}
    & \text{maximize } H(S) \\
    & \text{subject to } S \in \left(\Omega^{F}\right)^K.
\end{aligned}
\end{equation}

Note that the domain is a set of strings whose values under the function $H$ are permutation invariant, and that each sensor location can only be used once. If a total of $N$ lattice points in $\Omega^{F}$ are feasible for sensor placement, then there are $\binom{N}{K}$ possible placements. Exhaustive search becomes computationally intractable when $N$ is large, and therefore we apply the greedy algorithm to obtain an approximate solution in polynomial time. Using \thref{stringsubmodlemma} below, we can conclude that $H$ is monotone submodular in the case where all decay rates are equal. 
Submodularity guarantees that assumption $\mathbf{A_2}$ holds. Checking that the domain of $H$ is a uniform matroid of rank $K$ is straightforward, and by Lemma 2.2 of \cite{conforti1984submodular}, we see that assumption $\mathbf{A_1}$ also holds. Therefore, we can apply our performance bound to this example.  \\

\noindent \textbf{Nonhomogeneous Sensors:} We assume that all the sensors are working independently, but now the decay rate can change depending on the epoch in which the sensor is placed. The sensing capability of all unused sensors will become worse with each passing epoch, and once the sensors are placed, their decay rates will remain constant from their placement time onward. Specifically, we model the decay rate variation as $\lambda_{i} = \lambda_{1} + \zeta t_{i}$ for $i \in \{1, \ldots, K\}$, where $\{ t_{i} \}_{1\leq i \leq K}$ is a monotonically increasing time index sequence beginning at $0$, and $\zeta$ is the parameter controlling the increase in the decay rate. We can thus calculate the probability of detecting an occurring event at location $\mathbf{x} \in \Omega$ after placing $K$ nonhomogeneous sensors at locations $S$ using the formula 

\begin{equation}
    \begin{aligned}
    P(\mathbf{x},S) & = 1-\prod_{i=1}^{K}\left( 1-p_i(\mathbf{x},\mathbf{s}_i) \right) \\
    & = 1-\prod_{i=1}^{K}\left( 1- e^{-\lambda_{i} \|\mathbf{x}-\mathbf{s}_i\|} \right).
    \end{aligned}
\end{equation}

\noindent In the same manner as the previous homogeneous case, we can now phrase our optimization problem as
\begin{equation}
\begin{aligned}
\label{obj_fun_sensor_string}
    & \text{maximize } H(S) \\
    & \text{subject to } S \in \left(\Omega^{F}\right)^K.
\end{aligned}
\end{equation}

Here, the objective function has the same form as that of the previous example. Each sensor location can be used more than once so that assumption $\mathbf{A_1}$ is satisfied. However, our $H$ here is no longer permutation invariant because choosing the same location to place a sensor at different epochs results in different detection performance. Besides, the following lemma states that $H$ is monotone string submodular when the sensing capability of all unused sensors becomes worse, and hence assumption $\mathbf{A_2}$ is satisfied. 

\begin{lemma}\thlabel{stringsubmodlemma}
    If $\{ \lambda_{i} \}_{1\leq i \leq K}$ is a nondecreasing sequence, then $H$ is a monotone string submodular function . 
\end{lemma}

\begin{proof}
    Let $M$ and $Q$ denote two strings with $M \preccurlyeq Q$. Then $M$ and $Q$ can be written as: 
    $$
    \begin{aligned}
     M  = \mathbf{s}_{1}\cdots \mathbf{s}_{m}, \;
     P  = \mathbf{s}_{1}\cdots \mathbf{s}_{m} \mathbf{s}_{m+1}\cdots \mathbf{s}_{q}
    \end{aligned}
    $$
    with $1 \leq m < q \leq K-1$.
    
    Define $\Gamma(r,s,\mathbf{x}) = \prod_{i=r}^{s} ( 1- e^{-\lambda_{i} \|\mathbf{x}-\mathbf{s}_i\|} )$. Then 
    $0 < \Gamma (r,s,\mathbf{x}) \leq 1$ according to the above assumptions, and 
    $$
    \begin{aligned}
        H(P) - H(Q) & = \int R(\mathbf{x})\Gamma(1,m,\mathbf{x})(1-\Gamma(m+1,q,\mathbf{x})) d\mathbf{x} \\
        & \geq 0. 
    \end{aligned}
    $$
    
    \noindent Hence, $H(P) \geq H(Q)$. 

    Assuming $\mathbf{s}_{k} \in \Omega^{F}$,
    $$
    \begin{aligned}
    & H(M\mathbf{s}_{k}) - H(M) = \int R(\mathbf{x})\Gamma(1,m,\mathbf{x}) e^{-\lambda_{m+1} \|\mathbf{x}-\mathbf{s}_k\|} d\mathbf{x} \\
    & \stackrel{(a)}{\geq} \int R(\mathbf{x})\Gamma(1,q,\mathbf{x})e^{-\lambda_{q+1} \|\mathbf{x}-\mathbf{s}_k\|} d\mathbf{x} = H(P\mathbf{s}_{k}) - H(P). 
    \end{aligned}
    $$
    Inequality $(a)$ holds because $\Gamma(1,m,\mathbf{x}) \geq \Gamma(1,q,\mathbf{x})$ and $ \lambda_{m+1} \leq \lambda_{q+1}$. 

    Therefore, $H$ is a monotone string submodular function since we have shown that it satisfies the forward monotone and diminishing return properties. 
\end{proof} 

~\\
\noindent \textbf{Experiment:} We run the following experiments for both the homogeneous and nonhomogeneous cases. In our experiments, we consider a rectangular mission space $\Omega \subset \mathbb{R}^2$ of size $50 \times 40$, wherein we will  place $K$ sensors. The whole mission space is partitioned into four identical regions as shown in Fig.~\ref{sensors2}. The top right and bottom left regions have high event densities whose $R(\mathbf{x})$ for each $\mathbf{x}$ is randomly generated from $\text{Unif}(0.5,0.8)$. The other two regions have low event densities whose $R(\mathbf{x})$ for each $\mathbf{x}$ is randomly generated from $\text{Unif}(0.1,0.3)$. The above setting implies that the random events are more likely to occur in the top right and bottom left regions. For the nonhomogeneous case, we set $\zeta = 0.1$ and $t_{i} = 0.1(i-1)$ for $i \in \{1, \ldots, K\}$.

\begin{figure}[hbt!]
    \centering
    \includegraphics[width=0.9\columnwidth]{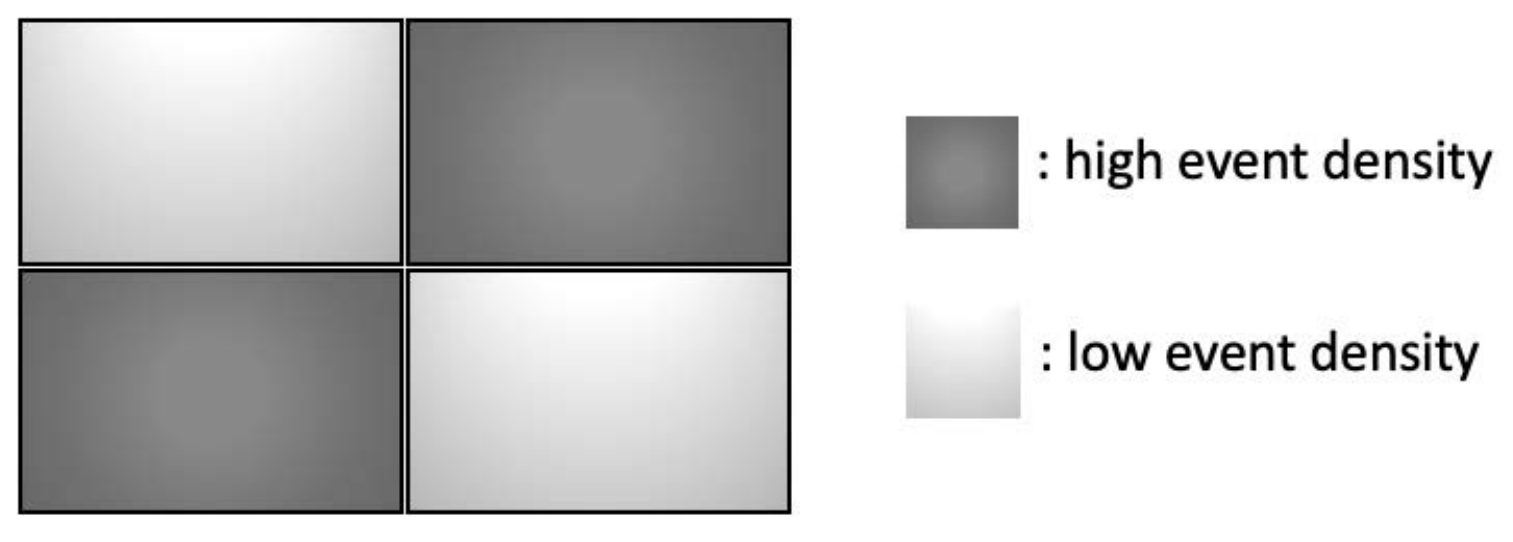}
    \caption{Mission space partition.} 
    \label{sensors2}
\end{figure}

We denote the $(1 - e^{-1})$ performance bound in \cite{nemhauser1978analysis} and \cite{streeter2008online} by $\beta_{0}$, the greedy curvature bound in \cite{conforti1984submodular} and its generalized version by $\beta_{1}$, and our performance bound by $\beta_{2} = f(G_K)/B_s$. A comparison of these performance bounds under different (initial) decay rates with $K=5$ is shown in Fig.~\ref{sensors_bound_disc}. Observe that in Fig.~\ref{sensors_bound_disc}, $\beta_{2}$ (red line) always exceeds $\beta_{1}$ (blue line), illustrating \thref{superioritytheorem}. In both cases we can observe instances where the $\beta_{2}$ is larger than $\beta_{0}
$, while $\beta_1$ is below $\beta_{0}$. When the (initial) decay rate is larger than 1, $\beta_{2}$ is close to 1. 

A comparison of these performance bounds under different number of placed sensors with $\lambda = 1$ and $\lambda_{1} = 1$ is shown in Fig.~\ref{sensors_bound_disc_K}. Observe that $\beta_{2}$ (red line) still dominates $\beta_{1}$ (blue line) with significant advantages. Both $\beta_{1}$ and $\beta_{2}$ decrease as we place more number of sensors. Both of these experiments demonstrate the advantages of our bound over those proposed in \cite{conforti1984submodular}. They also illustrate that for set and string submodular sensor coverage problems, employing a greedy strategy is a satisfactory choice when the sensors have weaker detection abilities. If sensor placement is allowed at any location within the mission space, the selections produced by the greedy strategy in the discrete version can serve as a good initial solution.

\begin{figure}[hbt!]
    \centering
    \includegraphics[width=0.98\columnwidth]{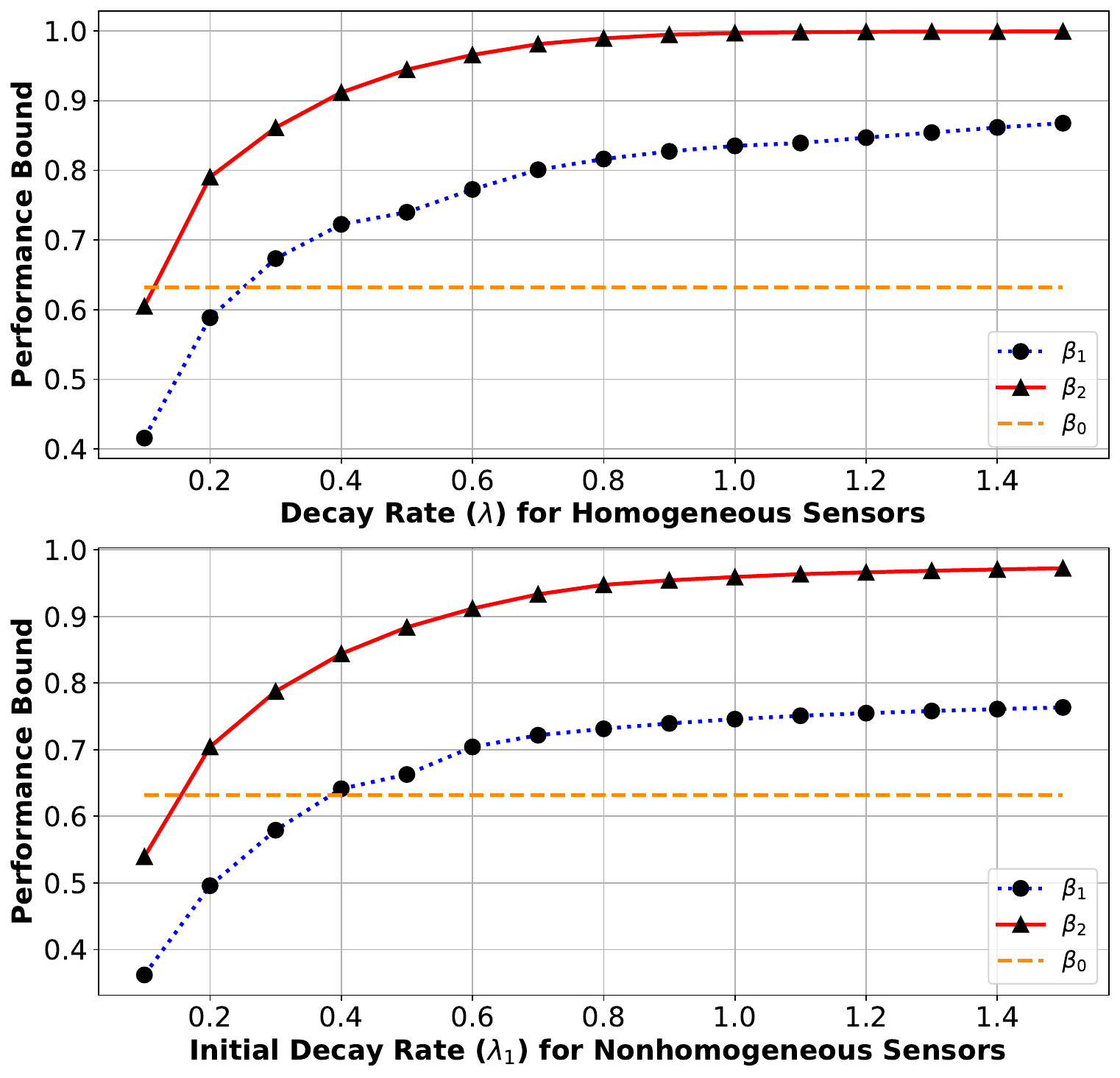}
    \caption{Performance bound comparison under different (initial) decay rates with number of placed sensors $K = 5$. 
    Upper Figure: Homogeneous Sensors; Lower Figure: Nonhomogeneous Sensors.}
    \label{sensors_bound_disc}
\end{figure}

\begin{figure}[hbt!]
    \centering
    \includegraphics[width=0.98\columnwidth]{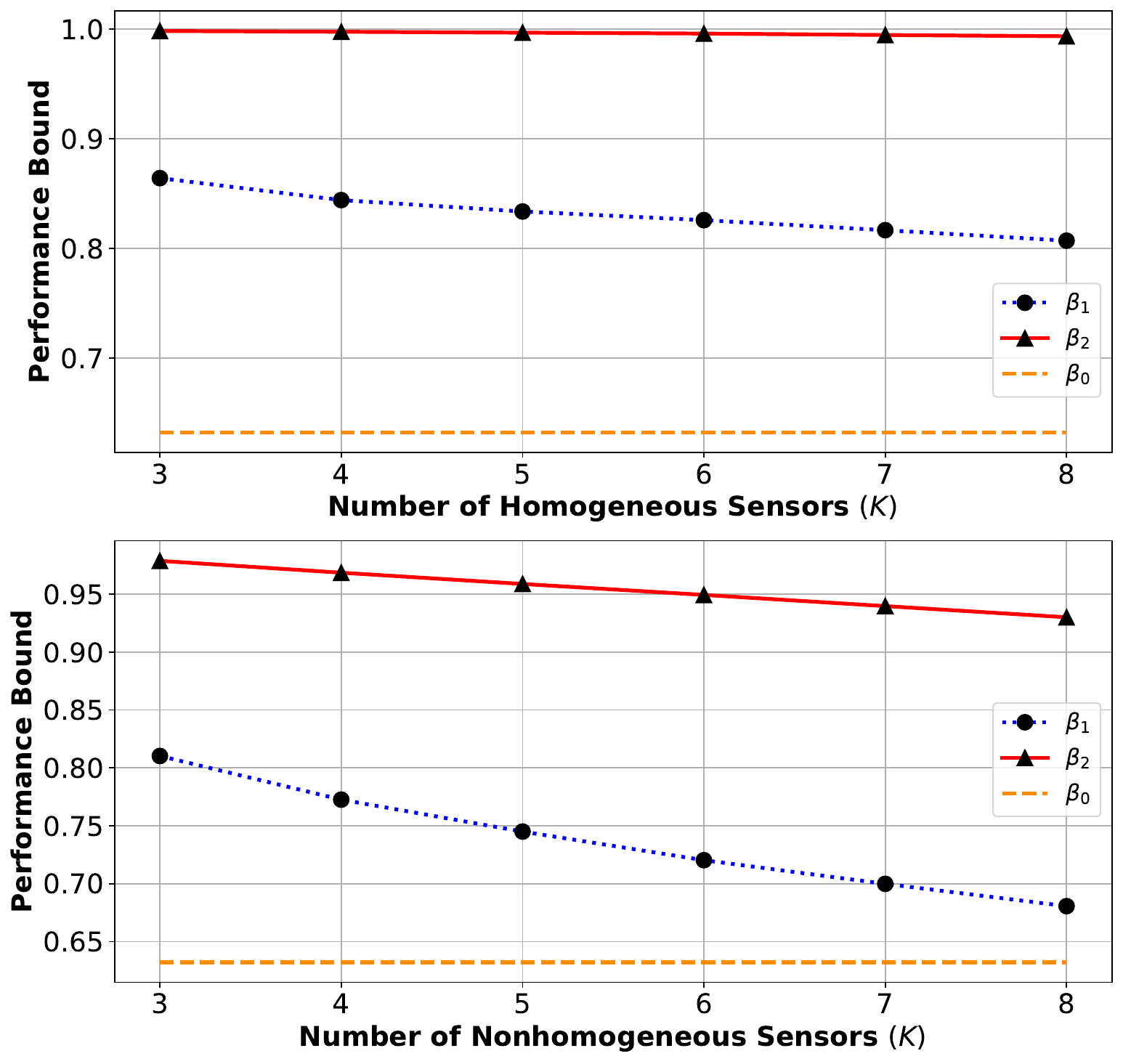}
    \caption{Performance bound comparison under different number of placed sensors. 
    Upper Figure: Homogeneous sensors with decay rate $\lambda = 1$; Lower Figure: Nonhomogeneous sensors with initial decay rate $\lambda_{1} = 1$.}
    \label{sensors_bound_disc_K}
\end{figure}

%%%%%%%%%%
\subsection{Social Welfare Maximization with Black-Box Utility Functions}
Suppose that a humanitarian organization wishes to distribute a fixed number of development grants to several remote villages. Each village can convert these grants into tangible local benefits, such as improved access to clean water or basic education, with the impact measured by a standardized ``community improvement index." The fund aims to allocate the grants in such a way as to maximize the aggregate sum of these community improvement indices across all recipient villages. This distribution must be accomplished without exceeding the total number of development grants available. Such a problem is an example of a social welfare maximization problem, which we will now define more generally. 

 In the welfare maximization problem, we are tasked with partitioning a set of $M$ items $I_M = \{i_1, \ldots, i_M\}$ among a set $A_N = \{a_1, \ldots, a_N\}$ of $N$ agents to maximize the social welfare function \[ \sum_{j=1}^N u_j(S_j). \]
 Let $S_j \subseteq I_M$ be the subset of items we assign to agent $a_j$, and $u_j$ is the utility function of agent $a_j$ defined on the power set of $I_M$, $u_j: 2^{I_{M}} \to \mathbb{R}_{\geq 0}$, which outputs the value of assigning set $S_j$ to agent $a_j$.
 
~\\
 \noindent \textbf{The Set Case:}
 Formally, we view this as a control problem wherein we have $M$ epochs at which to make choices. At each epoch we choose an element $(i_k, a_j) \in I_M \times A_N$, where we impose the rule that each $i_k \in I_M$ can be chosen \textit{at most } once. Such a rule guarantees that the set of all the elements we choose results in a partition of $I_M$ amongst the $N$ agents, and the set of all such subsets of $I_M \times A_N$ obtained from this process forms a matroid \cite{Lehmann2001}. Therefore, we seek to choose a set of pairs $\{(i_{k_1}, a_{j_1}), \ldots, (i_{k_M}, a_{j_M})\}$ that forms a partition of $I_M$ among the agents of $A_N$ for the following optimization problem:

\begin{equation}
\begin{aligned}
\label{obj_fun_welfare}
    & \text{maximize } \sum_{i=1}^Nu_i(S_i) \\
    & \text{subject to } \bigsqcup_{i=1}^NS_i = I_M,
\end{aligned}
\end{equation}
where the square union symbol $\bigsqcup$ refers to the disjoint union of all of the sets $S_i$.

To apply our results to the optimization problem, we need to verify that conditions $\mathbf{A_1}$ and $\mathbf{A_2}$ are satisfied. The fact that the domain forms a set matroid immediately gives us that $\mathbf{A_1}$ holds by Lemma 2.2 of \cite{conforti1984submodular}. To see that $\mathbf{A_2}$ holds, we need to impose the following condition on the utility functions of the agents. For any $j \in \{1,\ldots, N\}$ and any item $i_k \in I_M$, we have $u_j(\{i_k\}) \geq u_j(S_j \cup \{i_k\}) - u_j(S_j)$, where $S_j \subset I_M \setminus{ \{i_k\} }$. In other words, the initial value of assigning an item at epoch one serves as an upper bound for the incremental value obtained from adding that item to any set later on. Note that this condition is far weaker than submodularity, as the values obtained by assigning an item to an agent who has already received different items are allowed to fluctuate as long as they do not exceed the initial value of that item for that agent. From the aforementioned condition, we can see that $\mathbf{A_2}$ is satisfied by definition. We will assume that the utility functions are black-box functions whose values on elements are randomly generated in a manner to be described in the experiments section. 

~\\
\noindent \textbf{The String Case:}
To modify the aforementioned problem into a string optimization problem, we allow the utility functions of all agents to randomly change at each epoch. We again assume that the initial value of assigning each item at epoch one serves as an upper bound for the incremental value of receiving that item later on. Given that each item can be assigned at most once, we see that sum of the top $M$ initial values at epoch one still serves as an upper bound by \thref{applicationuse} even though we cannot guarantee that $o_i \in \bbS(G_{i-1})$ for all $i \in \{1, \ldots, M\}$.

\begin{figure}[hbt!]
    \centering
    \includegraphics[width=0.98\columnwidth]{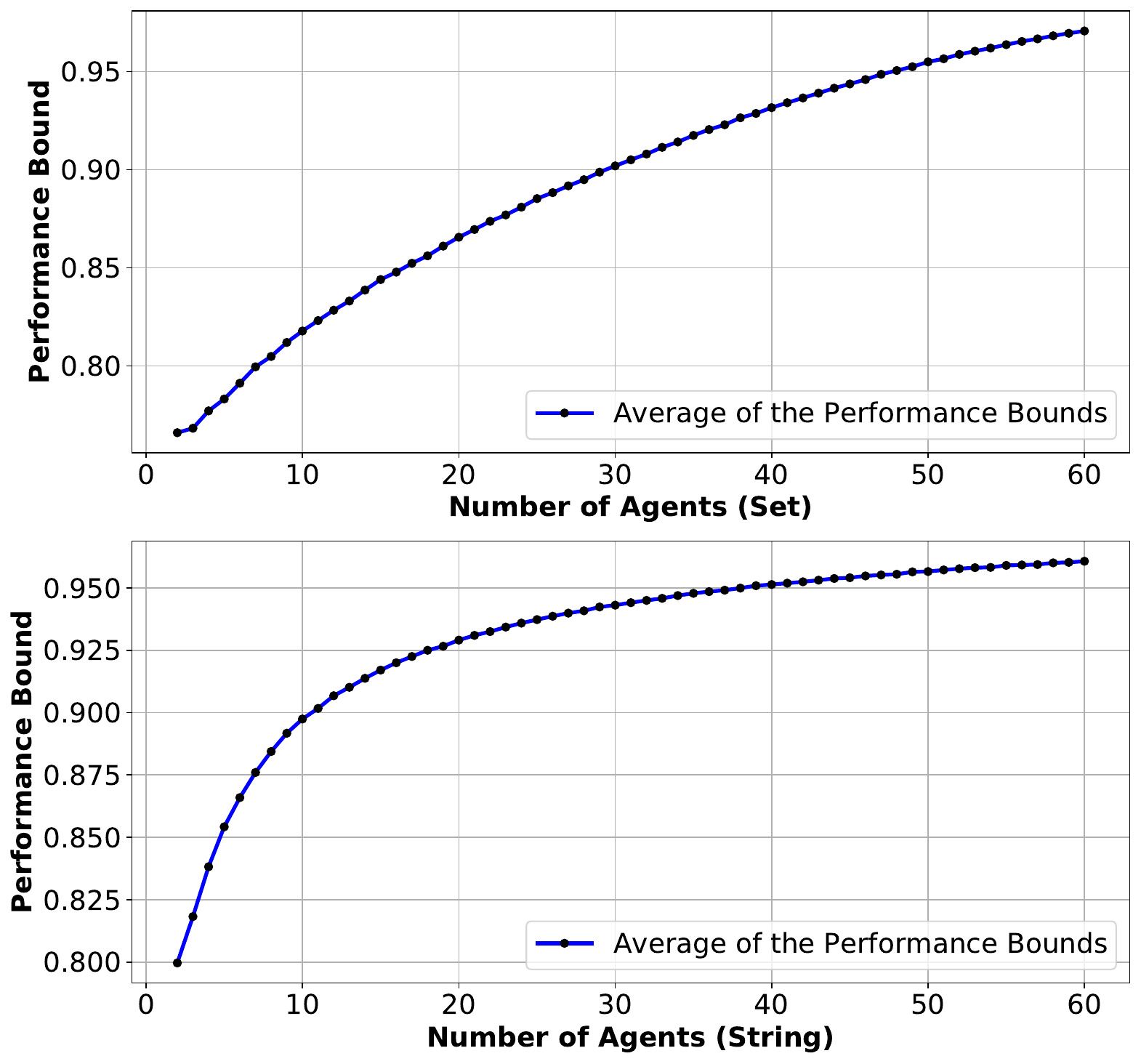}
    \caption{Performance bound comparison for set utility functions (upper figure) and string utility functions 
(lower figure). }
    \label{welfare_bound}
\end{figure}
~\\
\noindent \textbf{Experiment:} In our experiments, we have 60 items to distribute. Each utility function is a black-box function whose incremental gains for all $60$ items are randomly generated from $\text{Unif}(0,100)$. In Fig.~\ref{welfare_bound}, each point on the graph has its $x$-coordinate as the number of agents, and the $y$-coordinate as the average of the performance bounds of the greedy scheme after 1000 trials. A $95\%$ confidence interval for the average after 1000 trials is too small to be visible, and thus not included in the graph. 

We observe that as the number of agents increases, the performance bound of the greedy scheme also increases and trends towards optimal. We should note that the high performance bounds in our chosen examples are not anomalous, and instances where the performance bounds are less than $(1-e^{-1}) \approx 0.63$ are difficult to find. We also observe that the greedy scheme performs better in the string case than the set case. We believe this improved performance stems from the increased randomization of the values that are available to assign for each agent. This gives the greedy scheme more opportunity to choose larger values that are closer to the increments along the optimal string. Therefore, in situations where one is dealing with black-box utility functions with randomly generated values,  we could employ a greedy strategy since it is likely to be close to optimal under the conditions we described above.

\addtolength{\textheight}{-3cm}   % This command serves to balance the column lengths

\section{CONCLUSION AND FUTURE WORK}
We presented an easily computable lower bound with minimal assumptions for the performance of the greedy scheme relative to optimal scheme in string optimization problems. We then generalized two greedy curvature bounds in \cite{conforti1984submodular} and proved that our bound is superior to them. We also provided a counterexample that demonstrates that the $\alpha_G'$ greedy curvature bound in \cite{conforti1984submodular} is incorrect. We concluded with applications demonstrating the superiority of our result for monotone submodular set functions on matroids, an application to a monotone submodular string functions on string matroids, and lastly in problems where the objective function is not submodular and thus the bounds of \cite{conforti1984submodular} no longer applies. Future work would include applying our new bound to reinforcement learning problems and developing computable performance bound for distributed string optimization problems.

\section*{Appendix}
\printglossaries

\section*{References}
\bibliographystyle{ieeetr}
\bibliography{bib}

\begin{IEEEbiography}[{\includegraphics[width=1in,height=1.25in,clip,keepaspectratio]{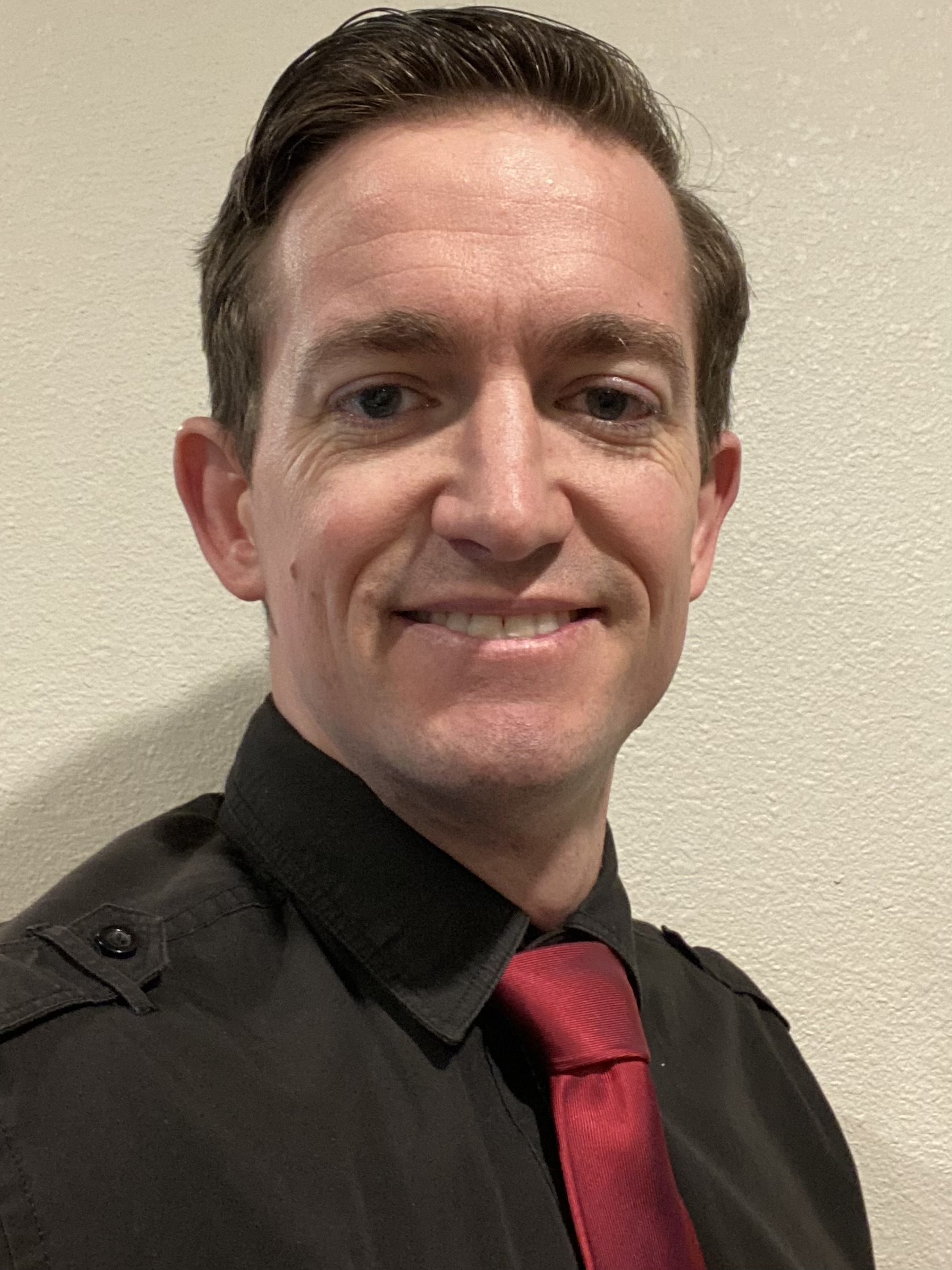}}]{Brandon Van Over} (Graduate Student Member, IEEE) received the B.S. degree in mathematics from the University of California, Berkeley in 2017, and the M.S. degree in mathematics from San Francisco State University in 2020. He is currently
pursuing a Ph.D. degree in electrical engineering
with Colorado State University, Fort Collins, CO,
USA. He placed first in the California State University research competition (2019), and during the same year received the Achievement Reward for College Scientists Scholarship (ARCS) of \$10,000. His research interests include optimization, machine learning, stochastic processes, and game theory.
\end{IEEEbiography}

\begin{IEEEbiography}
[{\includegraphics[width=1in,height=1.25in,clip,keepaspectratio]{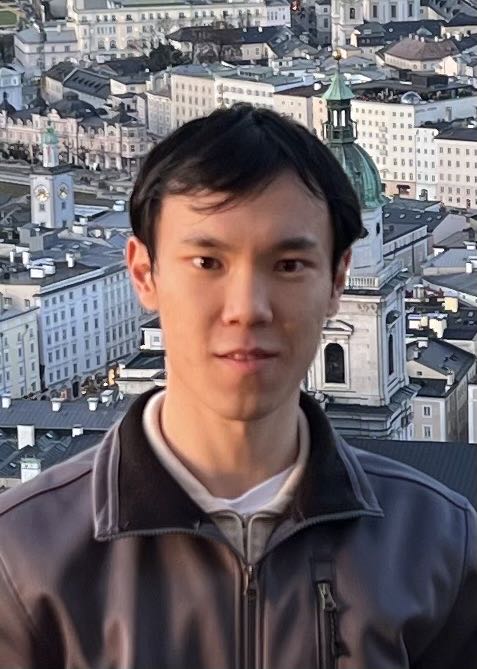}}]{Bowen Li} (Graduate Student Member, IEEE) received the B.S. degree in mathematics from Central China Normal University, Wuhan, China, and Colorado State University, Fort Collins, CO, USA, in 2018, through a dual-degree program. He received the M.S. degree in statistics from the University of Minnesota, Twin Cities, Minneapolis, MN, USA, in 2020. He is currently pursuing a Ph.D. degree in electrical engineering with Colorado State University, Fort Collins, CO,
USA. His research interests include statistical signal processing, optimization, and machine learning.
\end{IEEEbiography}

\begin{IEEEbiography}
[{\includegraphics[width=1in,height=1.25in,clip,keepaspectratio]{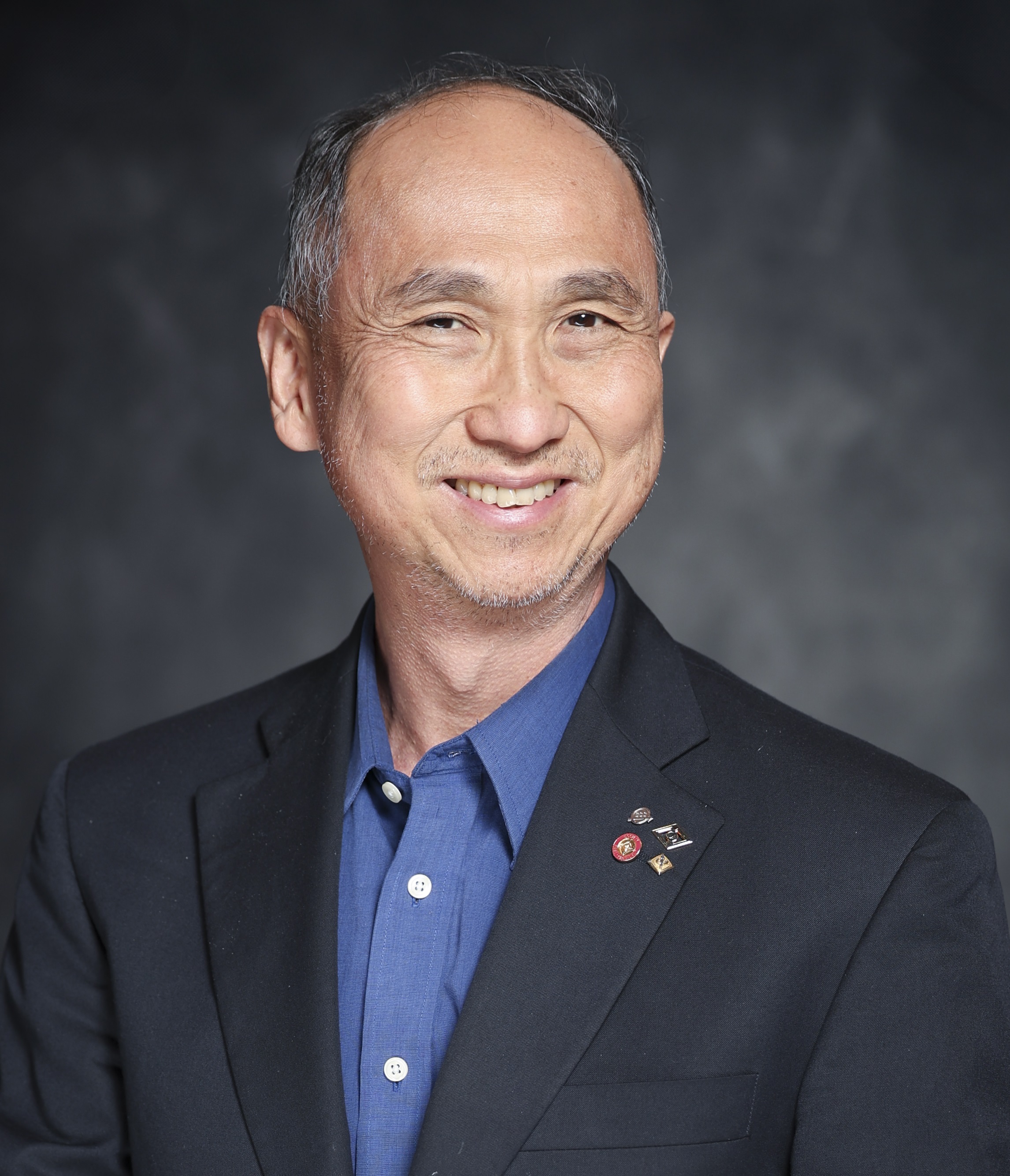}}]{Edwin K. P. Chong} (Fellow, IEEE) received the B.E.(Hons.) degree with First Class Honors from the University of Adelaide, South Australia, in 1987; and the M.A. and Ph.D. degrees in 1989 and 1991, respectively, both from Princeton University, where he held an IBM Fellowship. He joined the School of Electrical and Computer Engineering at Purdue University in 1991, where he was named a University Faculty Scholar in 1999. Since August 2001, he has been a Professor (and currently Head) of Electrical and Computer Engineering and Professor of Mathematics at Colorado State University. He coauthored the best-selling book, \emph{An Introduction to Optimization with Applications to Machine Learning} (5th Edition, Wiley, 2023). 

Prof.~Chong received the NSF CAREER Award in 1995 and the ASEE Frederick Emmons Terman Award in 1998. He was a co-recipient of the 2004 Best Paper Award for a paper in the journal Computer Networks. In 2010, he received the IEEE Control Systems Society Distinguished Member Award. He was the founding chairman of the IEEE Control Systems Society Technical Committee on Discrete Event Systems and served as an IEEE Control Systems Society Distinguished Lecturer. He was a Senior Editor of the IEEE Transactions on Automatic Control. He was the General Chair for the 2011 Joint 50th IEEE Conference on Decision and Control and European Control Conference. He has served as a member of the IEEE Control Systems Society Board of Governors and as Vice President for Financial Activities until 2014. He served as President in 2017.
\end{IEEEbiography}

\begin{IEEEbiography}
[{\includegraphics[width=1in,height=1.25in,clip,keepaspectratio]{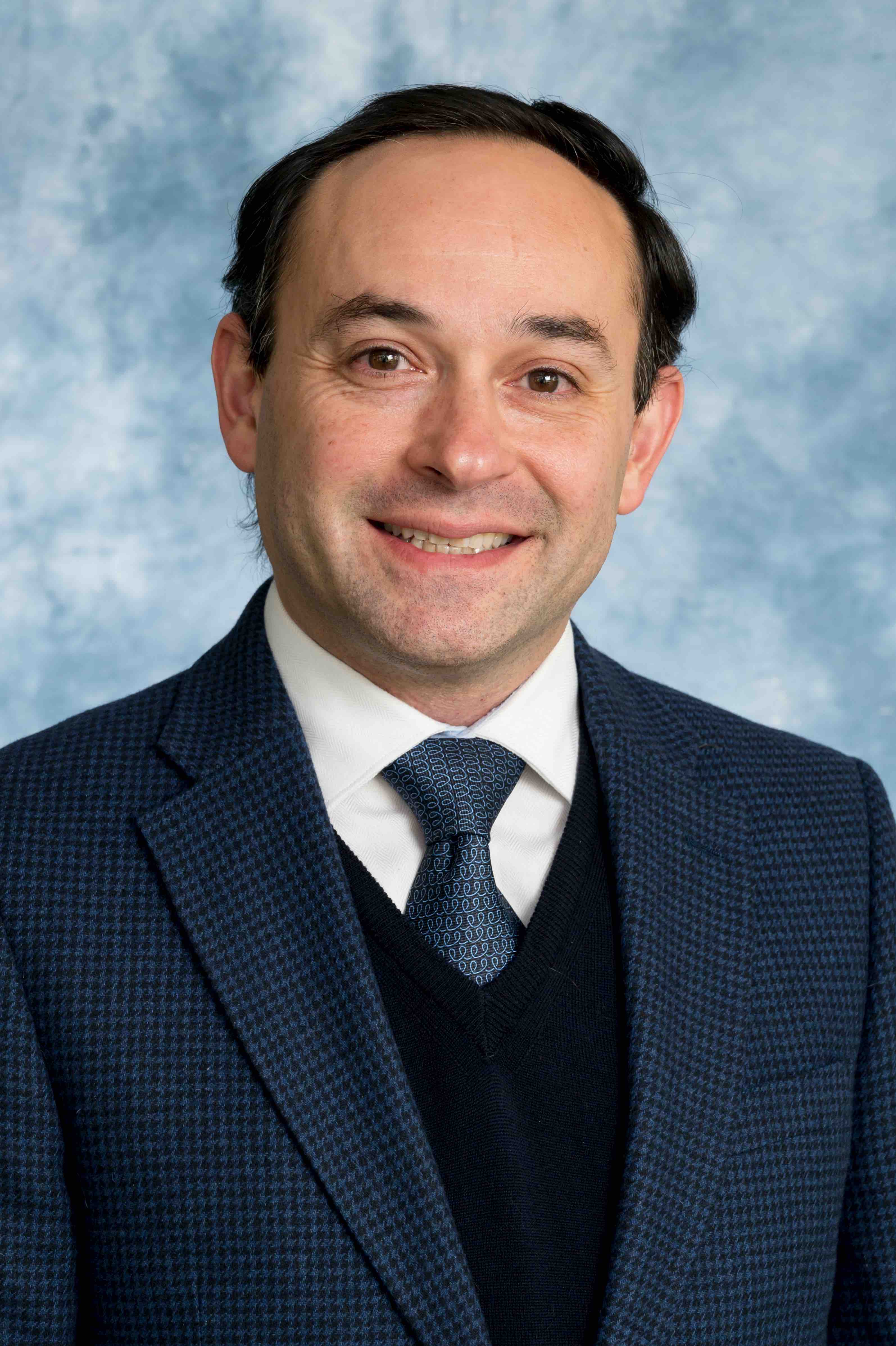}}]{Ali Pezeshki} (Senior Member, IEEE) received the B.Sc. and M.Sc. degrees in electrical engineering from the University of Tehran, Tehran, Iran,
in 1999 and 2001, respectively, and the Ph.D. degree in electrical engineering from Colorado State University, in 2004. In 2005, he was a Postdoctoral Research Associate with the Electrical and Computer Engineering Department, Colorado State University. From January 2006 to August 2008, he was a Postdoctoral Research Associate with the Program in Applied and Computational Mathematics,
Princeton University. In August 2008, he joined as a Faculty Member with Colorado State University, where he is currently a Professor with the Department of Electrical and Computer Engineering and the Department
of Mathematics. His research interests include statistical signal processing, machine learning, optimization, coding theory, geometry, applied harmonic
analysis, and bioimaging. He served on the Editorial Board for IEEE ACCESS from 2012 to 2018. 
\end{IEEEbiography}

\end{document}